\documentclass[submission,copyright,creativecommons]{eptcs}
\providecommand{\event}{NCMA 2026} %

\usepackage{iftex}

\ifpdf
  \usepackage{underscore}         %
  \usepackage[T1]{fontenc}        %
\else
  \usepackage{breakurl}           %
\fi

\title{Efficient Regex Matching with Sparse Counting-Sets}

\newcommand{\titlerunning}{Efficient Matching with Sparse Counting-Sets}
\newcommand{\authorrunning}{
    M. Berglund, B. van der Merwe and S. Sung
}

\hypersetup{
  bookmarksnumbered,
  pdftitle    = {\titlerunning},
  pdfauthor   = {\authorrunning},
  pdfsubject  = {EPTCS},               %
  pdfkeywords = {formal languages, regular expressions, regular expression engines, pattern matching, counter automata} %
}

\author{
  Martin Berglund
  \institute{Umeå University\\ Umeå, Sweden}
  \email{mbe@cs.umu.se}
  \and
  Brink~van~der~Merwe%
  \institute{
    Stellenbosch University\\
    Stellenbosch, South Africa
  }
  \email{abvdm@sun.ac.za}\\
  \and
  Sicheol Sung\footnote{Authors are listed alphabetically. While this work represents a collaborative effort, Sicheol Sung served as the lead author.}
  \institute{Yonsei University\\ Seoul, Republic of Korea}
  \email{sicheol.sung@yonsei.ac.kr}
}

\makeatletter
\newcommand*\bigcdot{\mathpalette\bigcdot@{.5}}
\newcommand*\bigcdot@[2]{\mathbin{\vcenter{\hbox{\scalebox{#2}{$\m@th#1\bullet$}}}}}
\makeatother

\usepackage{placeins}
\usepackage[table]{xcolor}
\usepackage[inline]{enumitem}
\usepackage[noEnd,indLines]{algpseudocodex}
\usepackage{algorithm}
\usepackage{amsfonts}
\usepackage{amsmath}
\usepackage{amssymb}
\usepackage{amsthm}
\usepackage{booktabs}
\usepackage{cite}
\usepackage{cleveref}
\usepackage{comment}
\usepackage{graphicx}
\usepackage{import}
\usepackage{mathrsfs}
\usepackage{mathtools}
\usepackage{microtype}
\usepackage{multirow}
\usepackage{pgfplots}
\usepackage{pifont}
\usepackage{subfig}
\usepackage{tabularx}
\usepackage{tikz}
\usepackage{tikzpagenodes}
\usepackage{url}
\usepackage{placeins}
\usepackage[commandnameprefix=always]{changes}

\pgfplotsset{compat=1.18}

\newcommand*{\regex}[1]{{\texttt{{#1}}}}

\newcommand*{\partialto}{{\,\rightharpoonup\,}}

\newcommand*{\init}{{\mathsf{init}}}

\newcommand*{\Guard}{{\mathsf{Guard}}}
\newcommand*{\Action}{{\mathsf{Action}}}
\newcommand*{\Assign}{{\mathsf{Assign}}}

\newcommand*{\nop}{{\mathsf{nop}}}
\newcommand*{\inc}{{\mathsf{inc}}}
\newcommand*{\reset}{{\mathsf{reset}}}
\newcommand*{\inact}{{\mathsf{inact}}}

\newcommand*{\opcheck}{{\mathsf{check}}}
\newcommand*{\addone}{{\mathsf{add}_1}}
\newcommand*{\merge}{{\mathsf{merge}}}

\newcommand*{\llist}[1]{{\langle {#1} \rangle}}

\newcommand*{\head}{{\mathsf{head}}}

\newcommand*{\sparseaddone}{{\mathsf{sparseAdd}_1}}
\newcommand*{\sparsemerge}{{\mathsf{sparseMerge}}}

\newcommand*{\To}{\Rightarrow}
\newcommand*{\powerset}[1]{{2^{#1}}}
\newcommand*{\N}{{\mathbb{N}}}

\newcommand*{\llb}{{[\![}}
\newcommand*{\rrb}{{]\!]}}
\newcommand*{\eval}[1]{{\llb{}{#1}\rrb{}}}
\newcommand*{\apply}[2]{{\eval{#2}_{#1}}}

\newcolumntype{Y}{>{\centering\arraybackslash}X}
\newcolumntype{Z}{>{\raggedleft\arraybackslash}X}

\theoremstyle{plain}

\newtheorem{theorem}{Theorem}
\newtheorem{lemma}[theorem]{Lemma}
\newtheorem{corollary}[theorem]{Corollary}
\newtheorem{proposition}[theorem]{Proposition}
\newtheorem{property}[theorem]{Property}

\theoremstyle{definition}
\newtheorem{definition}[theorem]{Definition}

\theoremstyle{remark}
\newtheorem{remark}[theorem]{Remark}
\newtheorem{example}[theorem]{Example}

\pgfplotsset{compat=1.18}
\usepgfplotslibrary{statistics}

\usetikzlibrary{patterns}
\usetikzlibrary{arrows.meta}
\usetikzlibrary{decorations.pathreplacing}
\usetikzlibrary{decorations.markings}
\usetikzlibrary{shapes.geometric,shapes.misc}
\usetikzlibrary{backgrounds,calc,fit,positioning}
\usetikzlibrary{automata, positioning, arrows}
\tikzset{
    state/.default={8mm},
    state/.style={%
        circle,%
        draw,%
        fill=white,
        align=center,
        inner sep=0pt,
        text width=#1%
    },
    >={Triangle[angle=45:4pt,fill=black]},
    accepting/.append style={double distance=2pt},
    initial text={}
}

\begin{document}

\maketitle

\begin{abstract}
    Regular expressions with counting operations~(c-regexes) offer a compact
    representation of repeating patterns by allowing numerical bounds to be
    added to subexpressions. Recent work introduced the counting-set data
    structure, which allows simultaneous updates of multiple counter values for
    efficient matching. However, this approach suffers from a performance
    bottleneck when counting-sets must be replicated due to the presence of
    branching transitions. We propose a sparse counting-set approach, which
    reduces the replication overhead by maintaining only essential counter
    values, thereby yielding a more efficient matching algorithm. 
\end{abstract}

\section{Introduction}

Counting operations make it possible to express bounded repetition, indicating
that something happens at least $l$ times but no more than $h$ times. Matching
strings using \emph{regular expressions with counting operations}~(c-regexes)
can be slow due to the additional complexity introduced by counting operations.
Two typical approaches are used for matching with c-regexes. The first is
\emph{counter-expansion}, employed by RE2~\cite{Cox10}, which rewrites c-regexes
into equivalent regexes without counting. However, this approach becomes
inefficient when the bounds of counting operations are large, as it
constructs duplicate states for each possible counter
value~\cite{TuronovaHLSVV20}. The second approach uses counter automata
(CAs)~\cite{GeladeGM12}; for example, variants of it are used in the Java regex
engine. This approach tracks the number of repetitions using counter variables,
but in the Java matcher, this approach inherits the inefficiencies of
backtracking matchers~\cite{TuronovaHHLVV22,DavisCSL18}.

During matching, a CA constructed from a c-regex~$r$ may generate a number of
counting configurations, linear in the bounds of the counting operators in $r$.
Thus, the worst-case matching time of Thompson-like matching algorithms, that
is, matching algorithms following a breadth-first approach, on CAs depends on
the bounds of the counting operators, which can be large. 
Since the bounds are
in practice encoded in decimal (as opposed to unary), a counter
range~$\{l, h\}$ contributes only~$O(\log h)$ characters to the textual length
of $r$, yet may induce $\Theta(h)$ distinct counter configurations during
matching. Consequently, the matching time can be exponential in the textual
length of $r$.

The \emph{counting-set} data structure~\cite{TuronovaHLSVV20,HolikSTV23} has
been proposed to address this inefficiency by allowing the simultaneous
application of the increment-by-one counter operation on multiple counter
values. However, the cost of replicating counting-sets at branching transitions
is expensive.
In this work, we address this problem by introducing a \emph{sparse
counting-set} structure, which reduces the size of counting-sets by keeping
track of only essential counter values. In summary, our contribution is not a
new automaton model, but a sparse representation for counting-sets with
preserved acceptance behavior and improved bounds on d-sparse flat c-regexes.

\paragraph{Related work.}

Compact handling of bounded repetition in pattern matching has a long history.
XFA~\cite{SmithEJ08} augments finite-state automata with auxiliary variables,
including counters, and is motivated by intrusion detection signatures with
heavily repeating subpatterns. Bj{\"{o}}rklund et al.~\cite{BjorklundMT15} study
the incremental evaluation of succinct regular expressions with numerical
constraints, giving bounds for maintaining matches under edits. In the
non-backtracking matching literature closer to our setting, counter
automata~\cite{GeladeGM12,TuronovaHLSVV20}, counting-set
automata~\cite{HolikSTV23,TuronovaHLSVV20}, and bit-vector
automata~\cite{GlaunecKM23} each trade expressiveness for efficiency
differently. Our work stays within the counting-set framework
of~\cite{HolikSTV23,TuronovaHLSVV20} and targets its main remaining bottleneck:
the cost of replicating large counting-sets at branching transitions, which we
eliminate through a sparse representation.

\section{Preliminaries}
\label{section:preliminaries}

By $\N$ we denote the set of positive integers and let $\N_0$ and $\N_{\infty}$
be $\N$ union the singleton sets $\{0\}$ and~$\{\infty\}$, respectively. We
assume that $n<\infty$ for all $n \in \N_0$. If $M, N \subseteq \N_0$ and~$k \in
\N_0$, then $k \circleddash M$, $N \circleddash k$ and $M \circleddash N$ denote
$\{k \circleddash m \mid m \in M\}$, $\{ n \circleddash k \mid n \in N \}$ and
$\{ m \circleddash n \mid m \in M, n \in N \}$, respectively, for $\circleddash$
any binary operator on $\N_0$. For $l \in \N_0$ and $h \in \N_\infty$, with
$l\le h$, we denote by $[l,h]$ the interval given by $\{n \in \N_0 \mid l \le n
\le h\}$. Thus, to simplify our presentation, we adopt the convention of denoting all integers greater than or equal to $l$, as $[l,\infty]$,
instead of $[l,\infty)$.

We denote a partial function~$f$ from a set~$X$ to a set~$Y$ by $f: X\partialto
Y$. A partial function~$f: X\partialto Y$ is often represented as~$\{x \!:
f(x)\}_{x \in X}$. For example, $\{a\!:x, b\!: y\}$ denotes a function~$f$ such
that $f(a) = x$, $f(b) = y$ and $f(x)$ is undefined for all other $x \in X$. A
relation~$\to$ on sets~$A$ and $B$ is \emph{functional} if for~$a \in A$ there
exists at most one $b \in B$ such that~$a \to b$.

An alphabet~$A$ is a set of characters, and a string~$w = a_1 a_2 \cdots a_n$
over $A$ is a finite sequence of characters with $a_i\in A$. The length~$|w|$ of
the string~$w$ is the number~$n$ of characters in $w$. Similar notation is
used to denote the cardinality of a finite set $B$, that is, $|B|$, but given
notational conventions, this will not create confusion. The empty string of
length zero is denoted by $\varepsilon$. For strings~$u = a_1 a_2 \cdots a_n$
and $v = b_1 b_2 \cdots b_m$, their concatenation~$u \cdot v$ is the string~$a_1
a_2 \cdots a_n b_1 b_2 \cdots b_m$. Languages over the alphabet~$A$ are sets of
strings over $A$. For two languages~$L_1$ and $L_2$, their concatenation $L_1
\cdot L_2$ is defined as $\{ uv \mid u \in L_1, v \in L_2\}$. For a
language~$L$, we define $L^0 = \{ \varepsilon \}$ and $L^{i+1} = L^i \cdot L$
for all $i \ge 0$. The Kleene star and plus of $L$ are defined by $L^* = \bigcup_{i \ge 0} L^i$ and
$L^+ = \bigcup_{i \ge 1} L^i$, respectively. Thus, $L^* = L^+\cup \{\varepsilon\}$.

\subsection{Symbolic regular expressions with counter operators}

In practice, regexes extensively use character classes, such as
\texttt{\textbackslash d}, which are equivalent to \texttt{(0|1|$\cdots$|9)} in
Perl-compatible regular expression syntax. That is, \texttt{\textbackslash d} is a predicate
matching a set of characters. To harmonize with this, we operate on symbolic
regular expressions in which counter operators might also be present. 
Readers more comfortable with ordinary regular expressions over an
alphabet~$A$ may, throughout, read each predicate~$\sigma$ simply as a single
character of~$A$; the algorithmic content does not depend on the predicate
machinery beyond the membership test $a \in \eval{\sigma}$.
We begin with the definition of a Boolean algebra---see~\cite{symlearn} for more details.

\begin{definition}
    An \emph{effective Boolean algebra} is a tuple $(A, \Sigma, \eval{\_}, \bot,
    \top, \lor, \land, \neg)$, where: (1)~$A$ is the \emph{domain}; (2)~$\Sigma$
    is a set of predicates over $A$ that is closed under the Boolean
    connectives~$\lor$, $\land$ and $\neg$, with~$\bot, \top \in \Sigma$;
    (3)~$\eval{\_}: \Sigma \rightarrow 2^{A}$ is a denotation function such that
    $\eval{\bot} = \emptyset$ and $\eval{\top} = A$, and for~$\varphi, \psi \in
    \Sigma$, $\eval{\psi \lor \varphi} = \eval{\psi} \cup \eval{\varphi}$,
    $\eval{\psi \land \varphi} = \eval{\psi} \cap \eval{\varphi}$, and
    $\eval{\neg \varphi} = A \setminus \eval{\varphi}$.
\end{definition}

When the domain~$A$ is an alphabet, each predicate~$\sigma$ defines the
set~$\eval{\sigma} \subseteq A$ of symbols satisfying the predicate, which are the
character classes of practical regular expressions.

\begin{definition}
    A \emph{symbolic regular expression with counting operations~(c-regex)} over
    $\Sigma$ is defined inductively as follows: (1)~$\varepsilon$ and $\sigma
    \in \Sigma$ are c-regexes; (2)~$r_1 \cdot r_2$ and $r_1 \mid r_2$ are
    c-regexes for c-regexes~$r_1$ and $r_2$; (3)~$r^*$, $r^+$ and $r^?$ are
    c-regexes for a c-regex~$r$; and (4)~for a c-regex~$r$ and $l \in \N_0$ and
    $h \in \N_{\infty}$, with $l\le h$, we also have that $r^{\{l, h\}}$ is a
    c-regex.
\end{definition}

By~$L(r)$ we denote the language associated with a c-regex~$r$, which is defined
inductively as follows. We let $\sigma\in\Sigma$ and~$r_1$ and $r_2$ be
c-regexes. Then, we define (1)~$L(\varepsilon) = \{ \varepsilon \}$ and
$L(\sigma) = \eval{\sigma}$;
(2)~$L(r_1 \cdot r_2) = L(r_1) \cdot L(r_2)$ and $L(r_1 | r_2) = L(r_1) \cup
L(r_2)$; (3)~$L(r_1^*) = L(r_1)^*$, $L(r_1^+) = L(r_1)^+$ and $L(r_1^?) = L(r_1)
\cup \{ \varepsilon \}$; and (4)~$L(r_1^{\{l,h\}}) = \bigcup_{i \in [l,h]}
(L(r_1))^i$.

We use $|r|_\Sigma$ to denote the number of occurrences of predicates in $r$.
The \emph{counting-height} of the c-regex~$r$ is the maximum depth of nested
counting operators in $r$. If $r$ contains no counter operations, we define the counting-height of $r$ as $0$.
A c-regex is \emph{flat} if the counting-height of
the regex is at most $1$. Let $H(r)$ denote the largest finite bound in a counter
range in $r$~(i.e.\ every subexpression $s^{\{l,h\}}$ has $l\le H(r)$ and,
if $h\not =\infty$, then $h\le H(r)$). 

We use a typewriter font to denote specific c-regex instances for readability.
For example, we write `\regex{(a|bc)\{2,10\}}' instead of `$(a | bc)^{\{2,
10\}}$'. Furthermore, although we may avoid infinite counter ranges by rewriting~$r^{\{l,\infty\}}$ as $r^{\{l,l\}}r^*$, we will see in Section~\ref{section:sparse-c-config}
that this rewrite harms matching performance when using sparse
counting-sets.

\subsection{Counter automata}

Counter automata (CAs) are an extension of nondeterministic finite-state
automata designed to match c-regexes~\cite{KongYCGHMY22,TuronovaHLSVV20}. Our
definition closely aligns with the formalism presented by Turonova et
al.\!~\cite{TuronovaHLSVV20}.

Let $C = \{c_1, c_2, \ldots, c_k \}$ be a finite set of \emph{counter
variables}. A function $\alpha: C \to \N_0$ is called an \emph{assignment} for
$C$, and the set of all assignments for $C$ is denoted as $\Assign(C)$. A
\emph{guard} over $C$ is any function~$\psi: c_i \in C \mapsto \psi_i$, where
$\psi_i$ is a Boolean combination of predicates of the form $\top$~(true),
$\bot$~(false), $c_i \le h$ and $l \le c_i$ with $l\in \N_0$ and $h\in\N_\infty$ (note, the predicate $c_i\le\infty$ is always satisfied). An assignment~$\alpha$ \emph{satisfies} $\psi$ if, for every
counter variable~$c_i$, the boolean expression~$\psi_i(n_i)$, which is obtained
from $\psi_i$ by substituting $c_i$ with its value $n_i = \alpha(c_i)$ in
$\psi_i$, evaluates to true; we denote this by $\alpha \models \psi$. The set of
all guards over $C$ is denoted by $\Guard(C)$.

An \emph{action} over $C$ is a function~$\theta: c_i \mapsto \theta_i(c_i)$,
where $\theta_i \in \{ \nop, \reset, \inc, \inact \}$, and (1)~$\nop: n \mapsto
n$, (2)~$\reset: n \mapsto 1$, (3)~$\inc: n \mapsto n+1$, and (4)~$\inact: n
\mapsto 0$. We denote by~$\apply{\theta}{\alpha}$ the assignment obtained by
applying $\theta$ to~$\alpha$. Thus, $\apply{\theta}{\alpha}(c_i) :=
\theta_i(\alpha(c_i))$. The set of all actions over $C$ is denoted as
$\Action(C)$.

\begin{example}
    Consider the assignment $\alpha = \{c_1: 2, c_2: 3\}$. We illustrate guards
    and actions using the following example.
    \begin{description}
        \item[Guards:] For guards $\psi = \{ c_1 : c_1 \le 2, c_2 : c_2 \ge 3
        \}$ and $\psi' = \{ c_1 : c_1 \ge 2, c_2 : c_2 \ge 4  \}$, we have that
        $\alpha \models \psi$ and $\alpha \not\models \psi'$. We implicitly
        conjoin the predicates in a guard; that is, we interpret a comma as a
        conjunction~$\land$.  
    
        \item[Actions:] For actions $\theta = \{ c_1: \nop, c_2: \reset \}$ and
        $\theta' = \{ c_1: \inc, c_2: \inact \}$, applying them to $\alpha$
        results in $\apply{\theta}{\alpha} = \{c_1: 2, c_2: 1\}$ and
        $\apply{\theta'}{\alpha} = \{c_1: 3, c_2: 0\}$.
    \end{description}
\end{example}

\begin{definition}
    A \emph{nondeterministic symbolic counter automaton}~(CA) is a tuple~$M =
    (Q, \Sigma, C, q_\init, \Delta, F)$, where (1)~$Q$ is a finite set of
    \emph{states}; (2)~$\Sigma$ is a finite set of predicates over an effective
    Boolean algebra; (3)~$C$ is a finite set of counter variables; (4)~$q_\init$
    is the \emph{initial state}; (5)~$\Delta \subseteq Q \times \Sigma \times
    \Guard(C) \times \Action(C) \times Q$ is the set of \emph{transitions}; and
    (6)~$F: Q \to \Guard(C)$ is the \emph{acceptance guard}.
\end{definition}
It may be useful to look ahead to Example~\ref{example:config-relation} with
Figure~\ref{fig:ca-construction} for a CA as we complete the definitions. For a
string~$w$, we define when a CA~$ M$ accepts~$ w$. A
\emph{configuration}~(config) of $M$ is a pair~$(q, \alpha) \in Q
\times \Assign(C)$. We define a \emph{transition relation}~$\to_{(M,a)}$, for a
character~$a \in A$, over configs, as follows:
\[
    (p, \alpha) \to_{(M,a)} (q, \beta) \iff \exists (p,\sigma, \psi, \theta, q)
    \in \Delta: \alpha \models \psi \land a \in \eval{\sigma} \land \beta =
    \apply{\theta}{\alpha}.
\]
Then, for a string~$w=a_1 a_2 \cdots a_n \in A^*$ with $a_i \in A$, we define
$\to_{(M,w)}$ as follows:
\[
    (q_0, \alpha_0) \to_{(M,w)} (q_n, \alpha_n)
    \iff (q_0, \alpha_0) \to_{(M,a_1)} (q_1, \alpha_1) \to_{(M,a_2)} \cdots
    \to_{(M,a_n)} (q_n, \alpha_n).
\]
We write $\to_w$ for $\to_{(M, w)}$ when $M$ is clear from the context.

The \emph{initial config} of $M$ is~$(q_\init, \alpha_\init)$, where
$\alpha_\init(c_i) = 0$ for all $c_i \in C$, and \emph{final configs} of $M$ are
configs~$(q, \alpha)$ satisfying $\alpha \models F(q)$. We call a config~$(q,
\alpha)$ \emph{reachable} if $(q_\init, \alpha_\init) \to_w (q, \alpha)$ for
some string~$w \in A^*$. If the config~$(q, \alpha)$ is final, then we say $M$
\emph{accepts} $w$. We define $L(M) := \{w \in A^* \mid \text{$M$ accepts
$w$}\}$.

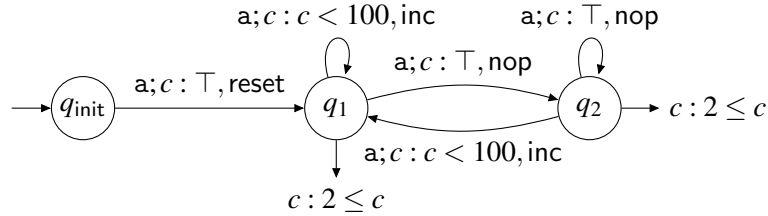
\begin{figure}
    \centering
    \begin{tikzpicture}
       \node[state, initial] (0) {$q_\init$};
       \node[state] (1) [right=2.5 of 0] {$q_1$};
       \node[align=center] (1') [below=0.5 of 1] {$c: 2 \le c$};
       \node[state] (2) [right=2.5 of 1] {$q_2$};
       \node[align=center] (2') [right=0.5 of 2] {$c: 2 \le c$};
    
       \path[->] 
        (0) edge [above] node
            {$\texttt{a}; c: \top, \reset$} (1)
        (1) edge [bend left=15, above] node
            {$\texttt{a}; c: \top, \nop$} (2)
        (2) edge [bend left=15, below] node
            {$\texttt{a}; c: c < 100, \inc$} (1)
        (1) edge [loop above] node
            {$\texttt{a}; c: c < 100, \inc$} ()
        (2) edge [loop above] node
            {$\texttt{a}; c: \top, \nop$} ()
        (2) edge (2')
        (1) edge (1')
        ;
    \end{tikzpicture}
    \caption{\label{fig:ca-construction}
        A CA equivalent to \regex{(aa*)\{2,100\}} with a single counter
        variable~$c$.
    }
\end{figure}

\begin{example}\label{example:config-relation}

Consider a c-regex~$r=\regex{(aa*)\{2,100\}}$. In reality this is equivalent to the expression $\regex{aaa*}$, but let us work out the counter use. For the CA~$M$ in
\Cref{fig:ca-construction}, $L(r) = L(M)$, that is, $M$ is equivalent to $r$.
The initial state always reads an $a$ and goes to $q_1$, which corresponds to
being `at' the $a^*$ in the expression. There the options are: reading another
$a$ by using the Kleene star~(going to $q_2$ without changing the count $c$);
reading another $a$ by reentering the counter operation~(incrementing $c$, only
possible if $c<100$); or accepting~(possible only if $c\ge 2$). Note that
$\text{aaa} \in L(r)$. The following are the sequences of transition relations
for $(q_\init, \{c: 0\}) \to_\texttt{aaa} (q, \alpha)$, for configs~$(q,
\alpha)$ reachable when reading~$\texttt{aaa}$. 
\begin{enumerate}
    \item $
        (q_\init, \{c: 0\})
        \to_\texttt{a} (q_1, \{c: 1\})
        \to_\texttt{a} (q_1, \{c: 2\})
        \to_\texttt{a} (q_1, \{c: 3\})
    $
    \item $
        (q_\init, \{c: 0\})
        \to_\texttt{a} (q_1, \{c: 1\})
        \to_\texttt{a} (q_1, \{c: 2\})
        \to_\texttt{a} (q_2, \{c: 2\})
    $
    \item $
        (q_\init, \{c: 0\})
        \to_\texttt{a} (q_1, \{c: 1\})
        \to_\texttt{a} (q_2, \{c: 1\})
        \to_\texttt{a} (q_1, \{c: 2\})
    $
    \item $
        (q_\init, \{c: 0\})
        \to_\texttt{a} (q_1, \{c: 1\})
        \to_\texttt{a} (q_2, \{c: 1\})
        \to_\texttt{a} (q_2, \{c: 1\})
    $
\end{enumerate}
We have that $\texttt{aaa} \in L(M)$ since $(q_1, \{c: 3\}), (q_2, \{c: 2\})$ and $(q_1, \{c:
2\})$ are final configs. 
\end{example}

\subsection{Thompson-like matching algorithm}

For a c-regex~$r$ and a string~$w \in A^*$, a Thompson-like matching algorithm
decides whether $w \in L(r)$, as follows: Consider a CA equivalent to $r$, which
we denote by~$M(r)$. A \emph{super-config} of $M$ refers to a set of configs of
$M(r)$. Next, we define the functional relations~$\To_a$ and $\To_w$, over
super-configs, for each super-config~$S$, character~$a \in A$ and $w = a_1 a_2
\cdots a_n \in A^*$, as follows:
\begin{align*}
    S \To_{(M, a)} \{
        (q, \beta) \in Q \times \Assign(C) \mid \exists (p, \alpha) \in S: (p,
        \alpha) \to_{(M, a)} (q, \beta)
    \}, \\
    (S \To_{(M, w)} S_n)
        \iff (S \To_{(M, a_1)} S_1 \To_{(M, a_2)} S_2 \To_{(M, a_3)} \cdots
        \To_{(M, a_n)} S_n).
\end{align*}
We write $\To_w$ for $\To_{(M, w)}$ when $M$ is clear from the context.

Note that if $\{ (p, \alpha) \} \To_w T$, then $(p, \alpha) \to_w (q, \beta)$ if
and only if $(q, \beta) \in T$. Finally, $w \in L(M)$ if $S_\init \To_w S_n$,
where $S_\init := \{(q_\init, \alpha_\init)\}$ is the initial super-config, and
$S_n$ contains a final config.

\begin{example}\label{example:super-config-relation}
    For the CA~$M$ in \Cref{fig:ca-construction}, we have the following
    derivation of super-configs.
    \begin{align*}
        \{ (q_\init, \{c: 0\}) \}
        &\To_\texttt{a} \{ (q_1, \{c: 1\}) \} \\
        &\To_\texttt{a} \{ (q_1, \{c: 2\}), (q_2, \{c: 1\}) \} \\
        &\To_\texttt{a} \{
            (q_1, \{c: 3\}), (q_2, \{c: 2\}), (q_1, \{c: 2\}), (q_2, \{c: 1\})
        \}.
    \end{align*}
    Note the correspondence between this derivation and the derivation sequences
    given in Example~\ref{example:config-relation}.
\end{example}

\subsection{Flat counter automata}

Next, we consider the \emph{position CA} constructed from a c-regex~$r$ with the
position construction as in~\cite{GeladeGM12}; for example,
\Cref{fig:position-construction} shows the position CA of the c-regex
\regex{a((bc)\{2,10\})+d\{3,3\}}. The position CA of $r$ has
$|r|_{\Sigma}$~states (where $|r|_{\Sigma}$ denotes the number of predicates
from $\Sigma$ in $r$) and at most ${(|r|_{\Sigma})}^2$ transitions, and it has the
properties outlined next~(see~\cite{HolikSTV23}). We begin with two required
definitions.

\begin{definition}[See~\cite{TuronovaHLSVV20}]
    For a counter variable~$c_i \in C$, the
    \emph{scope}~$Q_i \subseteq Q$ of $c_i$ is the smallest subset of $Q$ defined
    inductively such that $q\in Q_i$ if there exists a transition~$(p, \sigma,
    \psi, \theta, q) \in \Delta$ such that either (1)~$\theta_i = \reset$, or
    (2)~$p \in Q_i$ and $\theta_i \neq \inact$.
\end{definition}

For each state~$q \in Q$, we denote by $C_q := \{ c_i \in C \mid q \in Q_i \}$
the set of counter variables whose scope contains $q$. Furthermore, we define
$Q_0 := Q \setminus \bigcup_{c_i \in C} Q_i$ as the set of states outside the
scope of any counter. Consequently, $C_q = \emptyset$ for any~$q \in Q_0$.

\begin{definition}[Transition types]\label{definition:transition-types}
A transition type for a counter variable~$c_i \in C$ is a guard-action
pair~$(\psi(c_i), \theta(c_i))$ that falls into one of the following five
types:
\begin{enumerate*}[
    label=\textup{T\arabic*}:, start=0, ref=\textup{T\arabic*}
]
    \item \label{item:type-0} $(\top, \nop)$,
    \item \label{item:type-1} $(\top, \reset)$,
    \item \label{item:type-2} $(c_i < h_i, \inc)$,
    \item \label{item:type-3} $(l_i \le c_i, \reset)$, and
    \item \label{item:type-4} $(l_i \le c_i, \inact)$.
\end{enumerate*}
\end{definition}

The following two properties generalize the property presented in Appendix~B.1
of Hol{\'{\i}}k et al.~\cite{HolikSTV23}.

\begin{property}\label{property:transition}
    For any transition~$(p, \sigma, \psi, \theta, q) \in \Delta$ and counter
    variable~$c_i \in C$ of a position CA, the corresponding pair~$(\psi(c_i),
    \theta(c_i))$ must be one of the five transition types in
    Definition~\ref{definition:transition-types}, determined by the relation of state~$p$
    and~$q$ to the counter's scope~$Q_i$:
    \begin{enumerate}
        \item if $p \notin Q_i \land q \notin Q_i$ then $(\psi(c_i),
        \theta(c_i))$ is \ref{item:type-0};
        \item if $p \notin Q_i \land q \in Q_i$ then $(\psi(c_i), \theta(c_i))$
        is \ref{item:type-1};
        \item if $p \in Q_i \land q \in Q_i$ then $(\psi(c_i), \theta(c_i))$ is
        \ref{item:type-0}, \ref{item:type-2} or \ref{item:type-3}.
        \item if $p \in Q_i \land q \notin Q_i$ then $(\psi(c_i), \theta(c_i))$
        is \ref{item:type-4}.
    \end{enumerate}
\end{property}
Each of these follows naturally from the structure of the counter scopes. For example, in~1 we are outside the scope, so the counter has no effect,~\ref{item:type-0}; cases~2 and~3 enter and exit the scope. Only case~3 is somewhat complex, as the transition goes from a state inside the scope to another inside the scope. This can happen either without interacting with the counter operation at all, giving us~\ref{item:type-0}; by iterating the counter operation once, giving us~\ref{item:type-2}; or by exiting the counter operation and reentering it, using some \emph{enclosing} closure, resetting the counter in~\ref{item:type-3}.

\begin{property}\label{property:scope}
    For any config~$(q, \alpha)$ of the position CA and a counter variable~$c_i
    \in C$, $\alpha(c_i) \ne 0$ if and only if $q \in Q_i$.
\end{property}

We focus on the position CAs of flat c-regexes, referred to as \emph{flat
CAs}. In a flat CA, $|C_q| = 1$ for~$q \notin Q_0$ and $|C_q| = 0$ otherwise. 
In the flat case, we simplify the notation of a config~$(q, \alpha)$ as follows:
If $q \in Q_i$ for some $c_i \in C$, then the config is written as $(q,
\alpha(c_i))$; if $ q \in Q_0$, it is written as~$(q, 0)$. Note that every
c-regex can be converted into a flat c-regex using
\emph{counter-expansion}~\cite{Cox10} where we perform the expansion for all but
one of the nested counters.

\begin{figure}
    \centering
    \begin{tikzpicture}[]
        \node[state, initial] (0) {$q_\init$};
        \node[state] (1) [below=1.5 of 0] {$\texttt{a}_1$};
        \node[state] (2) [right=3.2 of 1] {$\texttt{b}_2$};
        \node[state] (3) [above=1.5 of 2] {$\texttt{c}_3$};
        \node[state] (4) [right=2.6 of 3] {$\texttt{d}_4$};
        \node (4') [below=of 4, align=left, xshift=1em] {
            $c_1: \top$\\
            $c_2: 3 \le c_2$
        };

        \node[xshift=-2cm,yshift=.1cm,align=left] (ml1) at (3) {
            $c_1: c_1 < 10, \inc$, \\
            $c_2: \top, \nop$, \\
            (\ref{item:type-2}, \ref{item:type-0})
        };

        \node[xshift=-2cm,yshift=-1.4cm,align=left] (ml2) at (3) {
            $c_1: 2 \le c_1, \reset$, \\
            $c_2: \top, \nop$, \\
            (\ref{item:type-3}, \ref{item:type-0})
        };
        
        \node[
            fit=(2) (3), draw, rounded corners=2mm, inner sep=2mm,
            label=above:$Q_1$
        ] (group box) {};
        
        \node[
            fit=(4), draw, rounded corners=2mm, inner sep=2mm,
            label=above:$Q_2$
        ] (group box) {};
    
        \path[->]
        (0) edge node[left, align=left] {
            $c_1: \top, \nop$, \\
            $c_2: \top, \nop$, \\
            (\ref{item:type-0}, \ref{item:type-0})
        } (1)
        (1) edge node[below=2pt, align=left] {
            $c_1: \top, \reset$, \\
            $c_2: \top, \nop$, \\
            (\ref{item:type-1}, \ref{item:type-0})
        } (2)
        (2) edge[bend right] node[right=4pt, align=left] {
            $c_1: \top, \nop$, \\
            $c_2: \top, \nop$, \\
            (\ref{item:type-0}, \ref{item:type-0})
        } (3)
        (3) edge node[above=1pt, xshift=3pt, align=left, xshift=-0pt] {
            $c_1: 2 \le c_1, \inact$, \\
            $c_2: \top, \reset$, \\
            (\ref{item:type-4}, \ref{item:type-1})
        } (4)
        (4) edge ($(4'.north)-(1em,0)$)
        (3) edge[bend right] node[pos=0.3,inner sep=0,outer sep=0] (midpe1) {} (2)
        (3) edge[bend right=10] node[pos=0.5,inner sep=0,outer sep=0] (midpe2) {} (2)
        (4) edge[loop right, align=left] node {
            $c_1: \top, \nop$, \\
            $c_2: c_2 < 3, \inc$, \\
            (\ref{item:type-0}, \ref{item:type-2});
        } (2);

        \draw[gray] ($(ml1.center)+(7pt,-7pt)$) edge[-o,out=0,in=125] ($(midpe1.center)+(1.7pt,0)$);
        \draw[gray] ($(ml2.center)+(7pt,-7pt)$) edge[-o,out=0,in=200] ($(midpe2.center)+(2.1pt,0)$);
    \end{tikzpicture}
    \caption{\label{fig:position-construction}
        The position CA of \regex{a((bc)\{2,10\})+d\{3,3\}}. Recall that
        \texttt{+} is the ``one or more'' operator. Counter scopes~$Q_1$ and
        $Q_2$ are denoted by rounded boxes. The symbol read by a transition (or
        rather its predicate) is indicated on the target state (once subscripts are ignored). As is usual for
        position automata, there is one state for each literal predicate in the
        expression (plus an additional initial state).
    }
\end{figure}
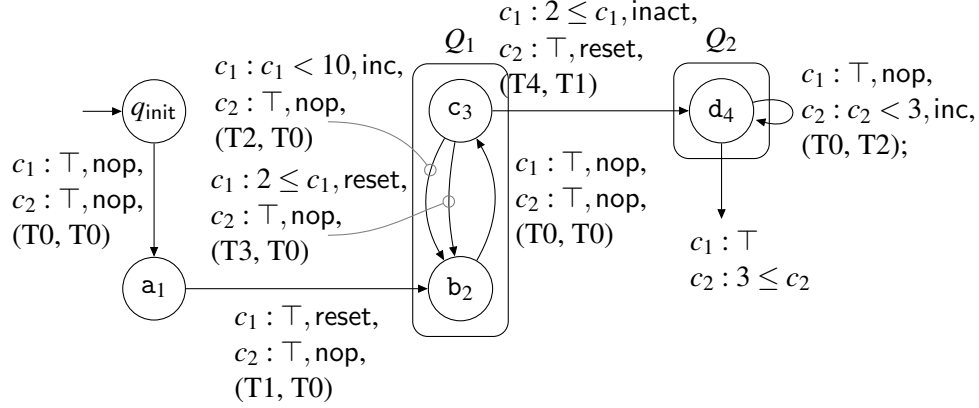

\subsection{Counter configurations}

Next, we describe the counting-set data structure that is used to represent a
set of counter values assigned to a counter variable at each state, following
Turo\v{n}ov\'{a}~et~al.~\cite{TuronovaHLSVV20}; see also
Le~Glaunec~et~al.~\cite{GlaunecKM23} for a related bit-parallel approach.

\begin{definition}\label{def:counting-set}
    A \emph{counting-set}, which is a representation based on an offset and an
    ordered list, is a tuple~$s = (o, \ell, l, h)$, where:
    \begin{itemize}
        \item $l \in \N_0$ and $h \in \N_\infty$ with $l \le h$ denote the
        interval~$[l, h]$ representing the required number of repetitions;
        
        \item $o \in \N_0$ is the offset; and
        
        \item $\ell = \llist{\ell_1, \ell_2, \ldots, \ell_n}$ with $|\ell| = n$
        is an ordered list of strictly~decreasing values with each~$\ell_i \in
        \N_0$ and $o - \ell_i \in [0, h]$.
    \end{itemize}
    We only consider counting-sets that can be obtained from the initial
    counting-set, i.e.\ $(0, \llist{0}, l, h)$ (representing the set $\{0\}$), using
    the operations $\inc$, $\addone$, $\merge$, defined below.
    
    We use $\ell_\head$ to denote the last element~$\ell_n$ of $\ell$; when
    $\ell$ is empty, we define $\ell_\head := \bot$. A counting-set~$s$
    represents the set~$N_s:=\{ o - \ell_i \mid 1 \le i \le n \} \subseteq [0,
    h]$. With $s=(o, \ell, l, h)$, we define the following operations on
    counting-sets, taking counting-sets as inputs to output a Boolean value for
    $\opcheck$ and a new counting-set for $\inc$, $\addone$ and $\merge$:
    \begin{itemize}
        \item $\opcheck(s) := [\ell_\head \ne \bot \land l \le o - \ell_\head]$; this checks whether $N_s \cap [l, h] \ne \emptyset$;
        
        \item $\inc(s) := (o + 1, \ell', l, h)$ represents $1 + N_s$; when $(o
        + 1) - \ell_\head \le h$ (i.e.\ if $\ell_\head$ is smaller than the
        upper bound $h$) $\ell' := \ell$, otherwise $\ell' := \llist{\ell_1,
        \ell_2, \ldots, \ell_{n-1}}$;
    
        \item $\addone(s) := (o, \ell', l, h)$ is defined when $o>0$, and
        represents $N_s \cup \{1\}$; we have $\ell' := \ell$ if $\ell_1 = (o-1)$,
        and $\ell' := \llist{o - 1, \ell_1, \ell_2, \ldots, \ell_n}$, otherwise;
        
        \item $\merge(s_1, s_2) := (o_1, \ell', l, h)$ represents the union of
        $s_1=(o_1,\ell^1,l,h)$ and $s_2=(o_2,\ell^2,l,h)$ for $o_1 \ge o_2$,
        where $\ell'$ is the result of merging $\ell_1$ and the list
        \[
            \llist{(o_1 - o_2) + \ell^2_1, (o_1 - o_2) + \ell^2_2, \ldots, (o_1
            - o_2) + \ell^2_\head},
        \]
        by also removing duplicates.
    \end{itemize}
\end{definition}

We can compute $\opcheck(s)$, $\inc(s)$, and $\addone(s)$ in time
$O(1)$. 
The operation $\merge(s_1, s_2)$ runs in $O(o_2)$ time because $o_2$ bounds (from above) both the size of $\ell^2$ and the number of elements in $\ell^1$ that must be processed during the merge.
The operation $\merge(s_1, s_2)$ is defined as a destructive update where $s_1$ consumes $s_2$. Consequently, $s_2$ becomes unavailable for further use.
When
computing~$\merge(s_1, s_2)$, the number of visited elements in $\ell_1$ and
$\ell_2$ is bounded by $2 \cdot o_2$. Moreover, since the offset of $s_2$ is
$o_2$, this implies that there were $o_2$ calls to $\inc$ on $s_2$ before
merging. 
Combined with the fact that $s_2$ is consumed by $\merge(s_1, s_2)$, this
implies that the amortized time complexity of $\merge(s_1, s_2)$ is
constant~\cite{HolikSTV23}. \Cref{fig:counting-set} illustrates the structure of
a counting-set. 

\begin{figure}
\centering
\includegraphics[width=0.8\linewidth,page=1]{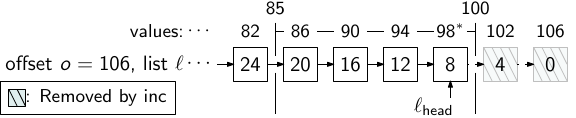}
\caption{A counting-set with $l = 85$ and $h = 100$.}
\label{fig:counting-set}
\end{figure}

Counting set automata use counting sets for efficient matching of c-regexes
with CAs~\cite{TuronovaHLSVV20,GlaunecKM23,HolikSTV23}. As a basis for our
proposed matching algorithm, we describe counter-configurations, which we
subsequently modify to obtain sparse counting-sets.

\begin{definition}\label{def:c-config}
For a flat CA we work with the simplified configs~$(q, n)$ introduced at the end of Section~\ref{section:preliminaries} (a flat CA has $|C_q| \le 1$ for
every state~$q$, so every assignment~$\alpha$ in a config reduces to a single value~$n = \alpha(c_i)$ if $q \in Q_i$, and to $n = 0$ if $q \in Q_0$). A \emph{counter configuration}~(c-config) is then a function~$f: Q \partialto\powerset{\N_0}$ that maps each state~$q$ to a set~$f(q)$ of counter values. A super-config~$S \subseteq Q \times \N_0$ is represented by the c-config~$f$
defined by $f(q) := \{ n \in \N_0 \mid (q, n) \in S \}$, with $f(q)$ undefined when no such pair occurs in~$S$.

\end{definition}

\begin{example}
The following c- and super-configs represent the same set of configurations.
\begin{enumerate}
\item[-] c-config~$f = \{ p\!: \{ 1, 2 \}, q\!: \{ 1, 2, 3 \},  r\!: \{ 1, 2, 3
\} \}$
\item[-] super-config~$S = \{ (p, 1), (p, 2), (q, 1), (q, 2), (q, 3), (r, 1), (r,
2), (r, 3) \}$
\end{enumerate}
\end{example}

\begin{definition}\label{def:c-config-transition}
Let $S$ and $T$ be super-configs of a flat CA such that $S \To_w T$ for some
string~$w \in A^*$. If $f$ and $g$ are c-configs that represent $S$ and $T$,
respectively, we write $f \To_w g$.
\end{definition}

Note that the counter values associated with a state~$q$ are placed in the
set~$f(q)$, which is represented using the counting-set structure. This
representation allows a c-config to increase all values simultaneously via the
operation~$\inc$.

\begin{figure}
\centering
\includegraphics[page=4,width=1.0\linewidth]{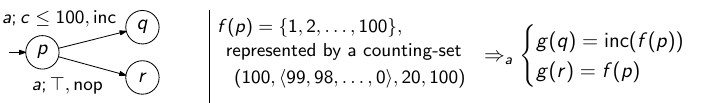}
\caption{\label{fig:replicating}
A schematic representation of the replication of a counting-set during a computation of the next
c-config~$g$ from a c-config~$f$. In this case, the counting-set~$f(p)$ is used to
compute both of the counting-sets~$g(q)$ and~$g(r)$.
}
\end{figure}

We define a flat c-regex~$r$ to be \emph{replicating} if its corresponding
position CA contains a state~$p \in Q_i$ with at least two outgoing transitions
whose operations~$\theta(c_i)$ are from $\{ \nop, \inc \}$ and whose
predicates~$\sigma_1$ and $\sigma_2$ are not disjoint (i.e., $\eval{\sigma_1}
\cap \eval{\sigma_2} \ne \emptyset$). \Cref{fig:replicating} illustrates a CA of
a replicating c-regex. In such a case, the matching algorithm in
\Cref{alg:counting-set} needs to replicate the counting-set~$s$ in order to
compute $s'$. This replication requires time linear in the size of $s$, which is
at most $H(r)$. Recall, $H(r)$ denotes the largest finite bound in a counter
range in $r$~(i.e., every subexpression $s^{\{l,h\}}$ has $l\le H(r)$ and, if
$h\not =\infty$, then $h\le H(r)$). Next, we use this observation on the cost of
replicating counting sets to analyse the time complexity of
\Cref{alg:counting-set}.  

\Cref{alg:counting-set} uses a counting-set data structure and c-configs to
match an input string against a CA obtained by applying the position automaton
construction to a flat c-regex. Notably, Line~\ref{line:condition-check} does
not explicitly check guards for $\nop$ and $\inc$. This is unnecessary because
(1)~the guard component that corresponds to $\nop$ and $\inc$ is always $\top$
and $c_i < h_i$ by Property~\ref{property:transition}, and (2)~counting-set's
$\inc$ operation implicitly filters out values that exceed the upper
bound~$h_i$.
For each state~$p$, the algorithm counts the number of outgoing transitions with
$\nop$ or $\inc$ operations in Line~\ref{line:replication-count}. If more than
one such transition exists, the corresponding counting-set is replicated in
Line~\ref{line:replication}. Finally, as guaranteed by
Property~\ref{property:transition}, the cases handled in
Lines~\ref{line:transition-0}, \ref{line:transition-1}, and
\ref{line:transition-2} are exhaustive and cover all possible transitions.
Although we will not provide a formal correctness argument, note that the natural correctness invariant for \Cref{alg:counting-set} is the following: After processing the first i characters of the input, the c-config $f$ maps each state $q$ to exactly the set of counter values $n$ such that the flat CA has a reachable config $(q,n)$ after that prefix.

For the proposition below to hold, we need to rewrite $r^{\{l, \infty\}}$ as $r^{\{l, l\}}r^*$.
Alternatively, one may provide special
treatment for the case $r^{\{l, \infty\}}$ (not given in the pseudocode listed for Algorithm~\ref{alg:counting-set}). For this, we split the transition for $(c_i < h, \inc)$ of T4 in Definition~\ref{definition:transition-types} into $(c_i < l, \inc)$ and $(c_i = l, \nop)$
if $h = \infty$, or by constraining the $\inc$ operation to have no effect once
a counter reaches its lower bound~$l$ (if $h = \infty$). 

\begin{proposition}[See~\cite{HolikSTV23}]\label{proposition:replicating}
    For a flat c-regex~$r$ and a string~$w$, \Cref{alg:counting-set} (modified for the case $r^{\{l, \infty\}}$) on the
    position CA of $r$ decides if $w \in L(r)$ in $O(H(r) \cdot
    (|r|_{\Sigma})^2 |w|)$-time. Furthermore, if $r$ is non-replicating, then~\Cref{alg:counting-set} runs in $O((|r|_{\Sigma})^2 |w|))$-time.
\end{proposition}

\begin{algorithm}
\caption{CA matching with counting-sets --- adapted from Holík et al.~\cite{HolikSTV23}.}
\label{alg:counting-set}
\begin{algorithmic}[1]
\State \textbf{Input:} position CA $M = (Q,\Sigma, C, q_\init, \Delta, F)$ of a
flat c-regex; input string $w = a_1 a_2 \cdots a_n \in A^*$
\State \textbf{Output:} \textbf{true} if $w \in L(M)$, \textbf{false} otherwise 
\Statex \textbf{Notation 1:} Because $M$ is flat, for every $q \in Q \setminus
Q_0$ there is a unique counter~$c_j \in C$ with $q \in Q_j$; we write $c(q) :=
c_j$ and $(l(q), h(q)) := (l_j, h_j)$ for its bounds, and set $c(q) := \bot$ for
$q \in Q_0$.
\Statex \textbf{Notation 2:} We write $\mathsf{replicate}(s)$ for a copy of the
counting-set~$s$, as depicted in \Cref{fig:replicating}.
\State $f \gets \{ q_\init: (0,\llist{0},0,0) \}$
\For{each character $a_i$ in $w$}
    \State $\textsf{counter\_ops} \gets \{ q: \emptyset \mid q \in Q \}$
    \State $\textsf{replications} \gets \{ p: 0 \mid p \in Q \}$
    \State $g \gets \{\}$
    \For{each key~$p$ of $f$ and transition~$(p, \sigma, \psi, \theta, q) \in \Delta$}
        \State $\theta' \gets \theta(c(q))$ if $c(q) \ne \bot$, and $\inact$ otherwise
        \If{\label{line:condition-check}
            $a_i \in \eval{\sigma}$ \textbf{and} $\left(\theta' \in \{ \nop, \inc \}~\textbf{or}~\opcheck(f(p))\right)$
        }
            \If{$\theta' \in \{\nop, \inc\}$}
                \State $\textsf{replications}(p) \gets \textsf{replications}(p) + 1$
                \label{line:replication-count}
            \EndIf
            \State $\textsf{counter\_ops}(q) \gets \textsf{counter\_ops}(q) \cup \{(p, \theta')\}$
        \EndIf
    \EndFor
    \For{each $q \in Q$ with $\textsf{counter\_ops}(q) \ne \emptyset$}
        \For{each $(p, \theta') \in \textsf{counter\_ops}(q)$}
            \If{\label{line:transition-0}
                $c(q) = \bot$ \textbf{and} $\theta' \in \{\inact, \nop\}$
            }
                \State $g(q) \gets (0, \llist{0}, 0, 0)$
            \ElsIf{\label{line:transition-1}
                $c(q) \ne \bot$ \textbf{and} $\theta' = \reset$
            }
                \State $g(q) \gets \addone(g(q))$ if $q$ is a key of $g$, and $\inc((0, \llist{0}, l(q), h(q)))$ otherwise
            \ElsIf{\label{line:transition-2}
                $c(q) \ne \bot$ \textbf{and} $\theta' \in \{ \nop, \inc \}$
            }
                \State $s \gets \mathsf{replicate}(f(p))$ if $\textsf{replications}(p) > 1$, and $f(p)$ otherwise
                \label{line:replication}
                \State $\textsf{replications}(p) \gets \textsf{replications}(p) - 1$
                \If{$\theta' = \inc$}
                    \State $s \gets \inc(s)$
                \EndIf
                \State $g(q) \gets \merge(g(q), s)$ if $q$ is a key of $g$, and $s$ otherwise
            \EndIf
        \EndFor
    \EndFor
    \State $f \gets \{q : g(q) \mid g(q) = (o, \ell, l, h), |\ell| > 0 \}$
\EndFor
\For{each $q \in F$ that is a key of $f$}
    \If{$\opcheck(f(q))$}
        \State \Return \textbf{true}
    \EndIf
\EndFor
\State \Return \textbf{false}
\end{algorithmic}
\end{algorithm}

\section{Sparse counting-sets}
\label{section:sparse-c-config}

Recall that the matching time complexity of replicating c-regexes depends on the
maximum size of counting-sets. We propose a new data structure, the \emph{sparse
counting-set}, which stores only \emph{useful} values. For example, for an
expression $E^{\{5,10\}}$ the sets $\{4,5,6,7,8,9,10\}$ and $\{4,10\}$ are
indistinguishable under the counting operators; they both satisfy $\opcheck$,
and do so when applying $\inc$ up to six times, but stop satisfying $\opcheck$
on a seventh $\inc$. This difference, at which any intermediary counter values
become useless, will be used throughout this section. We let $k=h-l+1$, with
$k=\infty$ when $h=\infty$.

\begin{definition}\label{def:spare-counting-set}
    A \emph{sparse counting-set} is a tuple~$s = (o, \ell, l, h)$ as in
    Definition~\ref{def:counting-set}, except it is produced from the initial
    counting-set using only the operations $\inc$, $\sparseaddone$ and
    $\sparsemerge$, defined next. Consequently, the CA matching algorithm with
    sparse counting-sets is exactly Algorithm~\ref{alg:counting-set}, where the
    standard operations $\addone$ and $\merge$ are replaced by their sparse
    counterparts $\sparseaddone$ and $\sparsemerge$, respectively.
    \begin{itemize}
        \item $\sparseaddone(s) := (o, \ell'', l, h)$, where $\ell''$ is
        constructed as follows. Let $\addone(s) = (o, \ell', l, h)$ with $\ell'
        = \llist{\ell'_1, \ell'_2 \ldots, \ell'_\head}$, noting that $\ell'_1 =
        o-1$. Then, the ordered list~$\ell''$ is obtained by removing~$\ell'_2$
        from $\ell'$ if and only if $\ell'_1 - \ell'_3 \le k$.

        \item $\sparsemerge(s_1, s_2) := (o_1, \ell'', l, h)$, where
        $\merge(s_1, s_2) = (o_1, \ell', l, h)$ and $\ell''$ is obtained from
        $\ell'$ by iteratively checking triples
        $\ell'_{i-1},\ell'_{i},\ell'_{i+1}$, in order of increasing $i$, and
        removing $\ell'_{i}$ if $\ell'_{i-1}-\ell'_{i+1}\le k$. Assuming $o_2\le
        o_1$, we may stop checking triples once $\ell'_{i-1}<(o_1-o_2)$, or once we
        have inspected all triples in $\ell'$.
    \end{itemize}
\end{definition}

As the definition implies, we will consider Algorithm~\ref{alg:counting-set} from here with the operations replaced by their sparse variants. The following lemmas are derived from the definitions of these operations.

\begin{lemma}\label{lemma:ineq-ell}
    A sparse counting-set~$s = (o, \ell, l, h)$ with $|\ell| \ge 3$, is such
    that the inequality~$\ell_i - \ell_{i+2} > k$ holds for any index~$i$.
\end{lemma}

\begin{proof}
    We show this by induction. A sparse counting-set must be obtained using the
    operations $\inc$, $\sparseaddone$ and $\sparsemerge$. Only $\sparseaddone$
    and $\sparsemerge$ lengthen a list $\ell$. The lemma holds for the initial
    counting-set. The operation $\sparseaddone$ either prepends an element
    (precisely when doing so does not violate this inequality) or replaces the
    leading element. As the offset is monotonically increasing under all
    operators, this new leading element must be greater than or equal to the
    element it replaces, preserving this inequality. The operation
    $\sparsemerge$ directly removes the elements that violate this inequality by
    scanning triples in the merged list, left to right. Note that after adding
    $(o_1-o_2)$ to each element in the second list before merging, all elements
    of this list will be between $(o_1-1)$ and $(o_1-o_2)$, inclusive, so we can
    stop checking triples once $\ell'_{i-1}<(o_1-o_2)$.
\end{proof}

Before stating the size bound, we record an invariant that simplifies the analysis. Although Definition~\ref{def:spare-counting-set} allows arbitrary sequences of $\inc$, $\sparseaddone$, and $\sparsemerge$, the sparse counting-sets that actually arise during a run of Algorithm~\ref{alg:counting-set} always satisfy $o - \ell_i \geq 1$ for every index $i$ (i.e., the value $0$ occurs only in the bare initial counting-set $(0,\langle 0\rangle,l,h)$, which has
$|\ell| = 1$).

The reason is as follows. $\sparsemerge$ is invoked only at line~25, in the branch $c(q) \neq \bot \wedge \theta' \in \{\nop,\inc\}$; by Property~\ref{property:transition} this forces $p, q \in Q_i$ for some counter $c_i$, and a routine induction on the outer iteration shows that any counting-set stored at a state in some $Q_i$ has an offset of at least $1$. Propagating this through the operations: $\mathsf{inc}$ strictly increments every $v_i := o - \ell_i$, $\sparseaddone$ only ever prepends an element with $v = 1$, and $\sparsemerge$ preserves the $v$-values of its operands. Since neither operand of a $\sparsemerge$ in Algorithm~\ref{def:spare-counting-set} is the bare initial counting-set, the invariant $v_i \geq 1$ therefore propagates throughout. We will use this fact freely in Lemma~\ref{lemma:sparseness} and Proposition~\ref{prop:merge-cost}.

\begin{lemma}\label{lemma:sparseness}
For a sparse counting-set $s = (o,\ell,l,h)$ with $h \geq 1$ arising during a run of Algorithm~1, the size $|\ell|$ is bounded by
\[
H'(l, h) \;:=\;
\begin{cases}
\displaystyle 2\left\lceil \dfrac{h}{k+1} \right\rceil & \text{if } h \neq \infty,\\
2 & \text{if } h = \infty.
\end{cases}
\]
\end{lemma}

\begin{proof}
For $h = \infty$ we have $k = \infty$, and $|\ell| \leq 2$ is forced by Lemma~\ref{lemma:ineq-ell}.

For $h \neq \infty$, let $\ell = \ell_1, \dots, \ell_n$ and set $v_i := o - \ell_i$, so $v_1 < v_2 < \cdots < v_n$. The case $n = 1$ is trivial since $2\lceil h/(k+1) \rceil \geq 2$. So assume $n \geq 2$; then $s$ is not the bare initial counting-set (which has $|\ell| = 1$), and by the invariant noted above, we have $v_i \geq 1$ for all $i$. Hence $v_i \in [1, h]$, and by Lemma~\ref{lemma:ineq-ell} we have $v_{i+2} - v_i \geq k + 1$.

Consider the odd-indexed subsequence $v_1 < v_3 < \cdots < v_{2j-1}$, where $j = \lceil n/2 \rceil$. Consecutive elements differ by at least $k + 1$, so
$v_{2j-1} \geq v_1 + (j-1)(k+1) \geq 1 + (j-1)(k+1)$. Combined with $v_{2j-1} \leq h$:
\[
(j-1)(k+1) \;\leq\; h - 1
\;\Longrightarrow\;
j \;\leq\; \left\lfloor \frac{h-1}{k+1} \right\rfloor + 1
       \;=\; \left\lceil \frac{h}{k+1} \right\rceil.
\]
The same argument applied to the even-indexed subsequence (whose smallest element is at least $2$, giving $(i-1)(k+1) \leq h - 2$) yields $\lfloor n/2 \rfloor \leq \lceil h/(k+1) \rceil$. Summing gives $n \leq 2\lceil h/(k+1) \rceil$. 
\end{proof}

From the early stopping criteria in $\sparsemerge$, and using that, we merge
sparse lists (i.e.\ we can apply the previous lemma to these two lists), we
obtain the following proposition, which guarantees that the time complexity of
$\sparsemerge$ improves on that of $\merge$ (assuming a large enough value of
$k$), while preserving amortized constant time in the non-replicating case.

\begin{proposition}\label{prop:merge-cost}
For sparse counting-sets $s_1, s_2$ arising during a run of Algorithm~\ref{alg:counting-set} with $o_2 \leq o_1$, the computation of $\mathsf{sparseMerge}(s_1, s_2)$ considers at most $2\lceil o_2/(k+1) \rceil$ elements from each of the two lists to achieve the merge, and then investigates at most $2 \cdot 2\lceil o_2/(k+1) \rceil$ triples from the merged list to restore sparsity.
\end{proposition}

\begin{figure}
    \centering
    \includegraphics[width=0.8\linewidth,page=2]{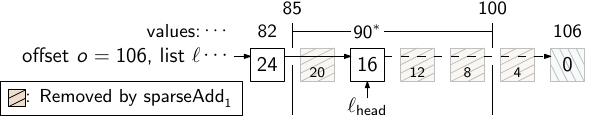}
    \caption{
        A \emph{sparse counting-set} with $l = 85$ and $h = 100$, where $k =
        16$. We assume the sparse counting-set is constructed by using sparse
        operations instead of the corresponding non-sparse operations as in
        Figure~\ref{fig:counting-set}.
    }
    \label{fig:sparse-counting-set}
\end{figure}

\begin{lemma}\label{lemma:sparse-counting-set}
    Let $N_s, N_{s'} \subseteq \N_0$ be sets of counter values represented by a
    counting-set~$s$ and corresponding (i.e.\ produced with the corresponding
    sequence of sparse operations) sparse counting-set~$s'$, respectively. Then,
    $N_{s'} \subseteq N_s$ and, for any~$n \in N_s \setminus N_{s'}$ with $n \in
    [l,h]$, there exists $n' \in N_{s'}$ such that $n' \in [l, h]$. That is,
    $\opcheck(s) = \opcheck(s')$.
\end{lemma}

\begin{proof}[Sketch]
    The tuples $s$ and $s'$ clearly agree in offset, $l$ and $h$. Let $\ell$ and
    $\ell'$ be the ordered lists in $s$ and $s'$ respectively. By construction,
    $\ell'$ is a subsequence of $\ell$. We proceed by contradiction. Assume that
    $\opcheck(s)\ne \opcheck(s')$. Then we must have $\opcheck(s)$ true and
    $\opcheck(s')$ false, as $s'$ represents a subset of $s$. Then there is some
    value $v$ in $\ell$ that satisfies the check (i.e.\ $l\le o-v\le h$), which
    does not occur in $\ell'$. The value $v$ must have been removed either by
    $\sparseaddone$ or by $\sparsemerge$, but each removes a value only if it
    would otherwise occur in a context $v',v,v''$ where $v'-v''<k+1$. However,
    we must then have $l\le o-v \le h$, $o-v'<l$ and $o-v''>h$, but then the
    difference between $v'$ and $v''$ would have to be at least $h-l+2$, which
    causes a contradiction.
\end{proof}

\begin{corollary}
    Let $f$ be the c-config~$f$ of a flat CA representing the initial
    super-config. Given a string $w \in A^*$, we compute a c-config~$g$ such
    that $f \To_w g$ using counting-set operations and replication of
    counting-sets. In parallel, we define another c-config~$g'$ obtained by
    applying the same procedure, but replacing each counting-set operation with
    its sparse counterpart. Then, $g$ represents a final super-config if and
    only if $g'$ represents a final super-config.
\end{corollary}

\begin{proof}
    Follows from Lemma~\ref{lemma:sparse-counting-set}, with each sparse
    counting-set in the configuration acting indistinguishably from its
    non-sparse counterpart.
\end{proof}

Combined with Lemma~\ref{lemma:sparseness}, we derive the following definition
of $d$-sparse. This property enables efficient matching with c-configs that use
sparse counting-sets instead of regular counting-sets, which we refer to as
\emph{sparse counter configurations}~(sc-configs).

\begin{definition}
    A flat c-regex $r$ is \emph{$d$-sparse} if $H'(l, h) \le d$ (with $H'(l,
    h)$ as defined in Lemma~\ref{lemma:sparseness}) for every subexpression of
    the form $s^{\{l, h\}}$ in $r$.
\end{definition}

\begin{example}\label{example:d-sparse}
    If all subexpressions of the form~$s^{\{l, h\}}$ of a c-regex~$r$ satisfy $l
    = 0$ or $h = \infty$, then the c-regex~$r$ is $2$-sparse.
\end{example}

With all this in hand, we can now state a bound on the membership problem using
sparse counting-sets, stated in terms of $d$. In many cases, $d$ will be quite
small, and in addition, in practice, this bound is pessimistic, as it accounts
for the possibility of replicating sets linearly many times.

\begin{theorem}\label{thm:sparse-matching-bound}
    Let $r$ be a $d$-sparse c-regex. For a given string~$w \in A^*$, we can
    decide whether $w \in L(r)$ in $O(d \cdot (|r|_{\Sigma})^2 |w|)$-time. If
    $r$ is non-replicating, then the algorithm runs in $O((|r|_{\Sigma})^2
    |w|))$-time.
\end{theorem}

\begin{proof}
    The automaton representing $r$ still has $(|r|_{\Sigma})^2$ as the number of
    transitions: While counting operations create additional transitions (e.g.,
    the two transitions from $c_3$ to $b_2$ in
    Figure~\ref{fig:position-construction}), $r$ being flat means any pair of
    states belongs to at most two counter scopes, making only a constant number
    of counter operations possible on transitions between that pair. The factor
    $d$ bounds the size of the representation of a sparse counting-set, allowing
    copying and/or merging sets within this bound. 
\end{proof}

\begin{remark}
    For the sake of a simpler presentation, we did not discuss an optimization
    where sparse counting-sets for ranges like $r^{\{0,h\}}$, $r^{\{1,h\}}$, or
    $r^{\{l, \infty\}}$ only need to store a single value.
\end{remark}

\section{Experimental evaluation}\label{section:experimental-evaluation}

We evaluate the efficiency of our sparse counting-set approach against
baselines. Efficiency is quantified by the total computation steps required for
symbol evaluations, guard checks, and counting-set operations. This includes
(1)~the number of processed counter values during $\merge$ operations and
(2)~the number of counter values involved in the replication of counting-sets.
Performance is assessed through targeted analyses of hand-crafted patterns and
large-scale evaluations on the Polyglot and Snort3 under adversarial and random
input scenarios.

\subsection{Results on targeted regular expressions}

Manually designed c-regexes have one of the following properties:
(a)~non-replicating, (b)~2-sparse, or (c)~neither. We evaluate and compare four
distinct methods---counter-expansion followed by Thompson NFA
matching~(c-expansion), Thompson-like matching with
super-configs~(super-config), standard counter-configs~(c-config), and
counter-configs with sparse counting-sets~(sc-config). For the c-expansion
method, the position CA construction is applied to the counter-expanded regex,
yielding a standard position NFA where all guards are $\psi(c) = \top$ and
actions are $\theta(c) = \nop$.

\begin{figure}[htb]
\begin{subfloat}[\regex{a*a\{k,k\}}]{
    \includegraphics[width=0.3\textwidth]{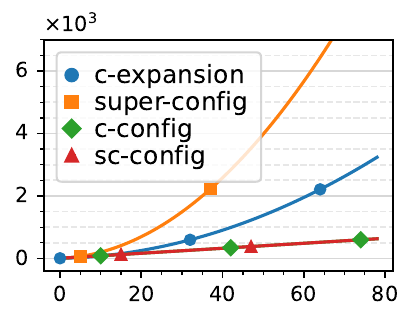}
    \label{fig:case-study-0}
}
\end{subfloat}
\hfill
\begin{subfloat}[\regex{(aa*)\{0,k\}}]{
    \includegraphics[width=0.3\textwidth]{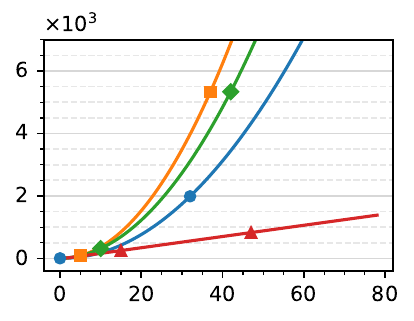}
    \label{fig:case-study-1}
}
\end{subfloat}
\hfill
\begin{subfloat}[\regex{a*(a|aab)\{k,k\}}]{
    \includegraphics[width=0.3\textwidth]{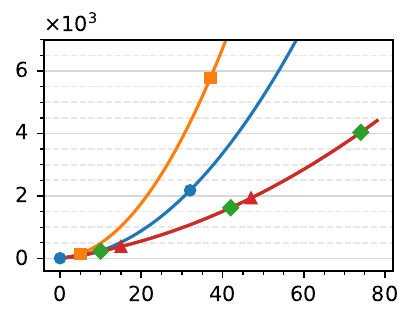}
    \label{fig:case-study-2}
}
\end{subfloat}
\caption{
    A comparison of the number of computation steps for our manually designed
    c-regexes, each matched against strings of the form~$\texttt{a}^\text{k}$.
    The x-axis indicates the length of the input string, which corresponds to
    the value~\texttt{k} in each pattern's counting operator.
}
\label{fig:case-study}
\end{figure}

The results in \Cref{fig:case-study} highlight the advantages of our sparse
counting-set approach. First, both the counter-expansion and super-config
methods fail to achieve linear-time matching. As shown in
\Cref{fig:case-study-1}, the standard c-config approach is also limited,
achieving linear-time performance only on non-replicating c-regexes. However, by
incorporating the sparse counting-set, we extend this capability to achieve
linear-time matching for replicating c-regexes that are $d$-sparse for a
bounded~$d$. The reason our sparse counting-set provides no advantage in
\Cref{fig:case-study-2} is explained by Lemma~\ref{lemma:sparseness}. For a
counting operator such as~\regex{\{k,k\}}, the upper bound~$H'(l,h)$ of the
number of values in a sparse counting-set becomes equal to an upper bound~$h$
for the non-sparse one, as $h - l + 1$ is simply $1$. As both bounds are
\texttt{k}, any potential benefit from sparsity is lost.

Furthermore, our experimental results demonstrate that for cases with a finite
upper bound~$h$ in a counter operation~$r^{\{l, h\}}$, the performance of the
super-config method is within a constant factor of the counter-expansion
approach. This result aligns with the theoretical explanation---each
config~$(q,n)$ in the super-config approach is equivalent to a state in the
position NFA generated by counter-expansion. The constant factor difference
arises from the elimination of guards and actions in the position NFA of the
expanded regex.

\subsection{Results on real-world regular expressions}

{\graphicspath{{figures/computations-comparison/}}

Next, we use c-regexes from two sources: the Polyglot regex corpus by Davis et
al.\!~\cite{DavisMCSL19} and the Talos LightSPD Snort3
ruleset~(\texttt{snapshot-31470}). Polyglot contains 511,196 patterns~(30,833
with counters) and Snort3 contains 9,755 patterns~(2,837 with explicit
counters). Table~\ref{tab:dataset-props} summarizes the structural properties of
these patterns, determined via static analysis of their position counting
automata.

We also compute the distribution of $H'(l,h)$ for 64,792 and 4,734~individual
counter operators found in flat regexes from Polyglot and Snort3, respectively,
to assess the practical bound of the parameter d on matching performance.
Figure~\ref{fig:dataset-dist-hprime} illustrates histograms of lower-bound~$l$
and upper-bound~$h$ of counter operators and the size~$H'(l,h)$ of sparse
counting-sets in our approach.

\begin{table}[htb]
\centering\small
\caption{Structural properties of patterns with counters.}
\label{tab:dataset-props}
\begin{tabular}{lrrrr}
\toprule
Property &
\multicolumn{2}{c}{Polyglot} &
\multicolumn{2}{c}{Snort3} \\
\cmidrule(r){1-1}
\cmidrule(r){2-3}
\cmidrule{4-5}
    Total with counters & 30,833 & (100\%) & 2,837 & (100\%) \\
    Flat~(depth~$d = 1$) & 29,911 & (97.0\%) & 2,749 & (96.9\%) \\
    Replicating & 319 & (1.0\%) & 79 & (2.8\%) \\
    2-sparse & 15,161 & (49.2\%) & 1,300 & (45.8\%) \\
    \bottomrule
\end{tabular}
\end{table}

\begin{figure}[htb]
\centering
\includegraphics[width=0.85\textwidth]{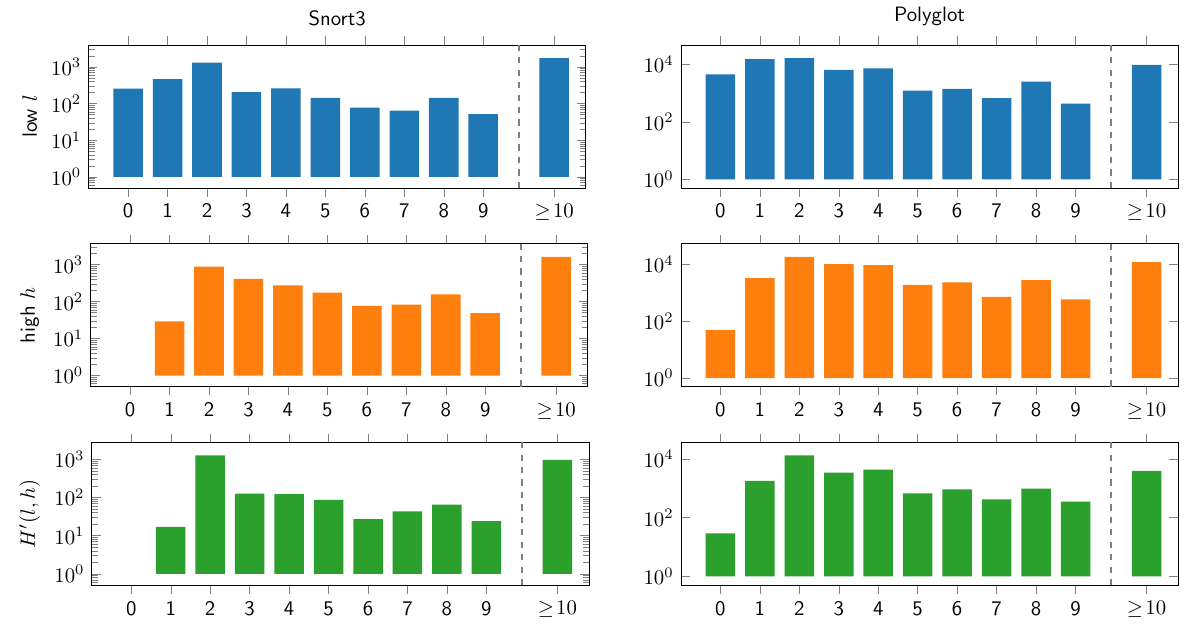}
\caption{
    Histograms of the lower bound~$l$, the upper bound~$h$, and the maximum
    sparse counting-set size~$H'(l,h)$ for each counter operator in flat
    c-regexes. All results are shown on a logarithmic scale.
}
\label{fig:dataset-dist-hprime}
\end{figure}

The bottom row of \Cref{fig:dataset-dist-hprime} is the empirical distribution
of $H'(l,h)$, whose per-pattern maximum is exactly the $d$ of
\Cref{thm:sparse-matching-bound}. In both corpora, the mass concentrates at
small values, with only a short tail at $H'(l,h) \ge 10$; this is driven by the
commonly occurring unbounded ranges, which give $H'(l,h) = 2$ by
Lemma~\ref{lemma:sparseness}. Hence, for most real-world patterns, $d$ is a
small constant, and the $O(d \cdot (|r|_\Sigma)^2 |w|)$ bound is within a small
factor of the non-replicating $O((|r|_\Sigma)^2 |w|)$.

Our experiments evaluate the robustness of our approach by conducting two
parallel sets of experiments simulating adversarial and random input scenarios.
In both cases, we simulated a partial matching scenario by wrapping each pattern
$r$ as \texttt{.*(r).*} to generate corresponding input strings. Our final
dataset is filtered to include only c-regexes that meet three criteria: they
(a)~contain at least one counting operation, (b)~are supported by our
implementation, and (c)~have a successfully generated test case.

\paragraph{Adversarial~(\Cref{fig:evilstrgen-experiments}):}

EvilStrGen~\cite{SuHLCG24} is designed to generate test cases that maximize the
number of matching steps. EvilStrGen was configured to assume a non-backtracking
engine and to generate string lengths of approximately 1,000. In summary, we use
10,108 and 611 c-regexes from the Polyglot and Snort3 datasets for the
experiments with EvilStrGen, respectively. 

\paragraph{Random~(\Cref{fig:xeger-experiments}):}

We simulate a random-input scenario by generating matched inputs with
Xeger\footnote{\url{https://pypi.org/project/xeger/}}. We use Xeger with its
default limit~$10$, which decides the maximum number of repetitions for Kleene
star, plus, and counting operators with infinite upper bounds. For the
experiments with Xeger, we use 10,141 and 620 c-regexes from the Polyglot and
Snort3 datasets, respectively.

\begin{figure}[h]
\textbf{Polyglot}

\subfloat[c-exp. vs. sup.-config]{
\includegraphics[width=0.3\textwidth]{
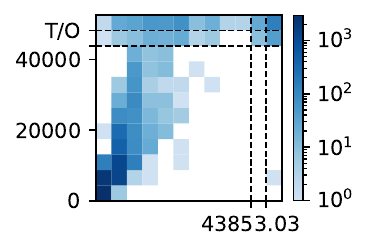}
}
\subfloat[c-exp. vs. c-config]{
\includegraphics[width=0.3\textwidth]{
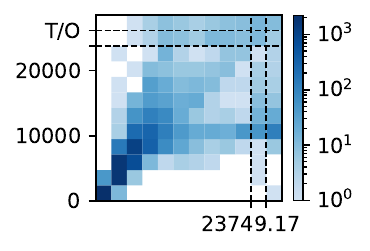}
}
\subfloat[c-exp. vs. sc-config]{
\includegraphics[width=0.3\textwidth]{
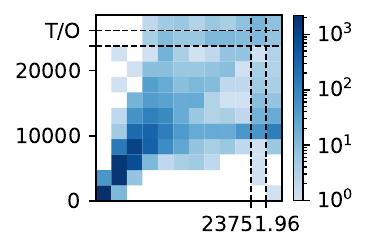}
}

\subfloat[sup.-config vs. c-config]{
\includegraphics[width=0.3\textwidth]{
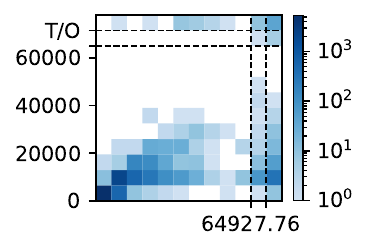}
}
\subfloat[sup.-config vs. sc-config]{
\includegraphics[width=0.3\textwidth]{
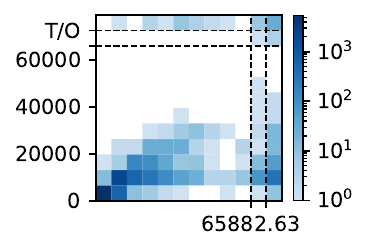}
}
\subfloat[c-config vs. sc-config]{
\includegraphics[width=0.3\textwidth]{
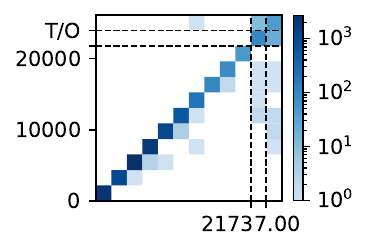}
}

\vspace{1.5em}

\textbf{Snort3}

\subfloat[c-exp. vs. sup.-config]{
\includegraphics[width=0.3\textwidth]{
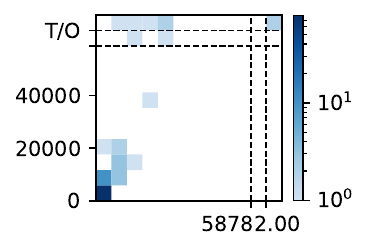}
}
\subfloat[c-exp. vs. c-config]{
\includegraphics[width=0.3\textwidth]{
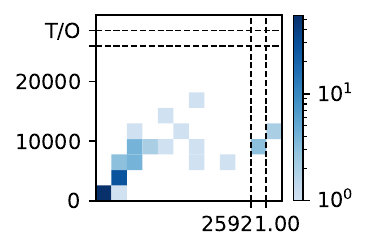}
}
\subfloat[c-exp. vs. sc-config]{
\includegraphics[width=0.3\textwidth]{
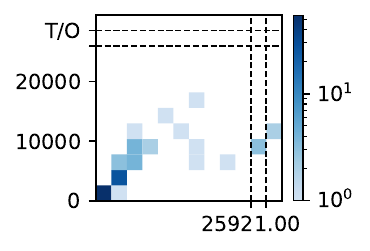}
}

\subfloat[sup.-config vs. c-config]{
\includegraphics[width=0.3\textwidth]{
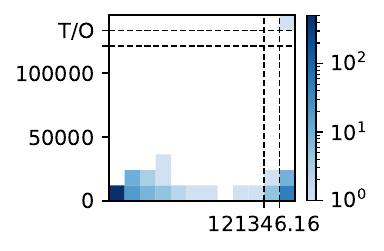}
}
\subfloat[sup.-config vs. sc-config]{
\includegraphics[width=0.3\textwidth]{
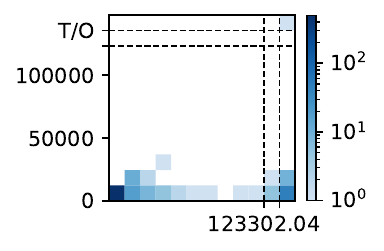}
}
\subfloat[c-config vs. sc-config]{
\includegraphics[width=0.3\textwidth]{
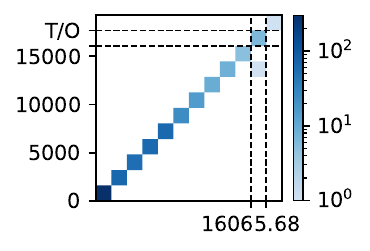}
}
\caption{\label{fig:evilstrgen-experiments}
    Comparisons of matching steps with input strings generated by EvilStrGen,
    presented as a sequence of 2D histograms. Each plot is labeled in a
    ``Method~1 vs. Method~2'' format, where the x-axis represents the step count
    for the first method and the y-axis for the second with the same scale. The
    color of each cell represents the number of such c-regex input string pairs.
    The second last line denotes the 99th percentile, and cases above (or to the
    right of) the ``T/O'' label timed out after one second.
}
\end{figure}

\begin{figure}[htb]

\textbf{Polyglot}

\subfloat[c-exp. vs. sup.-config]{
\includegraphics[width=0.3\textwidth]{
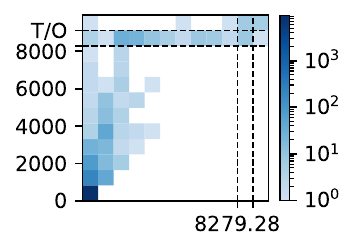}
}
\subfloat[c-exp. vs. c-config]{
\includegraphics[width=0.3\textwidth]{
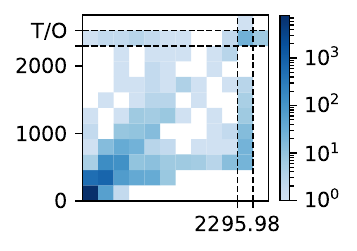}
}
\subfloat[c-exp. vs. sc-config]{
\includegraphics[width=0.3\textwidth]{
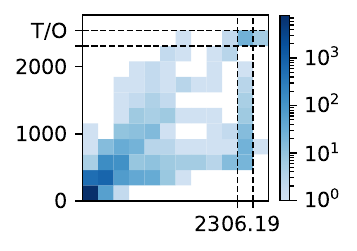}
}

\subfloat[sup.-config vs. c-config]{
\includegraphics[width=0.3\textwidth]{
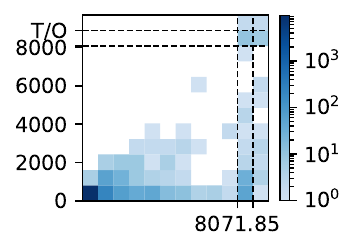}
}
\subfloat[sup.-config vs. sc-config]{
\includegraphics[width=0.3\textwidth]{
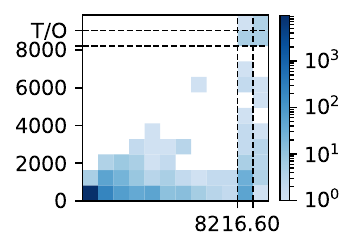}
}
\subfloat[c-config vs. sc-config]{
\includegraphics[width=0.3\textwidth]{
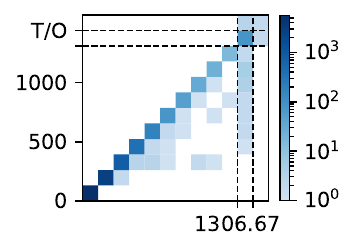}
}

\vspace{1.5em}

\textbf{Snort3}

\subfloat[c-exp. vs. sup.-config]{
\includegraphics[width=0.3\textwidth]{
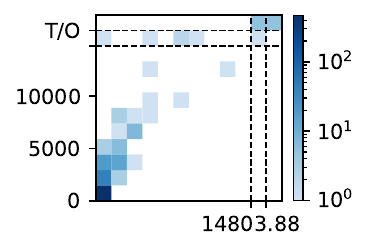}
}
\subfloat[c-exp. vs. c-config]{
\includegraphics[width=0.3\textwidth]{
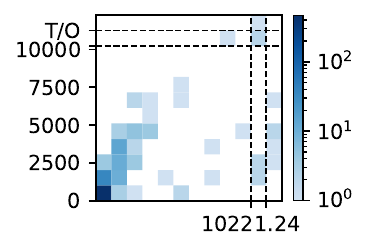}
}
\subfloat[c-exp. vs. sc-config]{
\includegraphics[width=0.3\textwidth]{
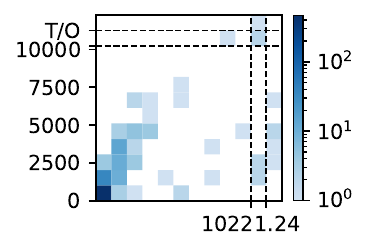}
}

\subfloat[sup.-config vs. c-config]{
\includegraphics[width=0.3\textwidth]{
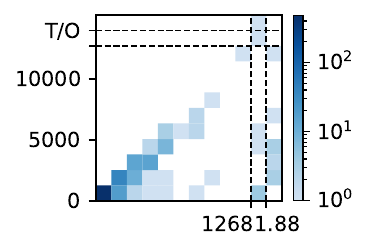}
}
\subfloat[sup.-config vs. sc-config]{
\includegraphics[width=0.3\textwidth]{
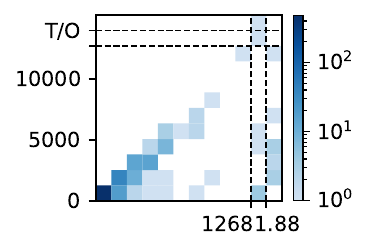}
}
\subfloat[c-config vs. sc-config]{
\includegraphics[width=0.3\textwidth]{
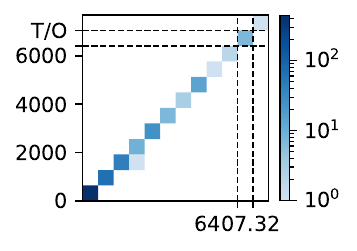}
}
\caption{\label{fig:xeger-experiments}
    Comparisons of matching steps with input strings generated by
    \texttt{xeger}, using the same experimental setup as in
    \Cref{fig:evilstrgen-experiments}. The 2D histograms show pairwise
    comparisons of the four methods, with cell color indicating the density of
    test cases.
}
\end{figure}

}

\FloatBarrier

\section{Conclusion and future work}
\label{section:conclusion}

We introduced c-configs, a formalization of on-the-fly matching with
counting-set automata~\cite{HolikSTV23}. Based on this formalization, we
presented a representation of super-configs as sparse counting-sets, minimizing
the number of counter values that are required to be tracked.

One can utilize counter expansion to convert non-flat c-regexes into flat
c-regexes. The choice of which counting operator to expand influences the
resulting expression. For instance, the c-regex~$(\regex{a\{2,3\}b)\{2,2\}}$ can
be expanded as either~$\regex{(aaa?b)\{2,2\}}$
or~$\regex{(a\{2,3\}b)(a\{2,3\}b)}$. Determining which expansion form leads to
more efficient CA-based matching is left as future research. We would also like
to explore the possibility of generalizing the (sparse) counting-set data
structure to support non-flat c-regexes.

\FloatBarrier

\subsection*{Acknowledgements}

The authors thank Rick Bower and River Martin for their detailed review of the
experimental results.

\nocite{*}
\bibliographystyle{eptcs}
\bibliography{references}

\begin{thebibliography}{10}
\providecommand{\bibitemdeclare}[2]{}
\providecommand{\surnamestart}{}
\providecommand{\surnameend}{}
\providecommand{\urlprefix}{Available at }
\providecommand{\url}[1]{\texttt{#1}}
\providecommand{\href}[2]{\texttt{#2}}
\providecommand{\urlalt}[2]{\href{#1}{#2}}
\providecommand{\doi}[1]{doi:\urlalt{https://doi.org/#1}{#1}}
\providecommand{\eprint}[1]{arXiv:\urlalt{https://arxiv.org/abs/#1}{#1}}
\providecommand{\bibinfo}[2]{#2}

\bibitemdeclare{inproceedings}{symlearn}
\bibitem{symlearn}
\bibinfo{author}{George \surnamestart Argyros\surnameend} \& \bibinfo{author}{Loris \surnamestart D'Antoni\surnameend} (\bibinfo{year}{2018}): \emph{\bibinfo{title}{The Learnability of Symbolic Automata}}.
\newblock In: {\slshape \bibinfo{booktitle}{Computer Aided Verification - 30th International Conference, {CAV} 2018, Proceedings, Part {I}}}, {\slshape \bibinfo{series}{Lecture Notes in Computer Science}} \bibinfo{volume}{10981}, \bibinfo{publisher}{Springer}, pp. \bibinfo{pages}{427--445}, \doi{10.1007/978-3-319-96145-3\_23}.

\bibitemdeclare{inproceedings}{BjorklundMT15}
\bibitem{BjorklundMT15}
\bibinfo{author}{Henrik \surnamestart Bj{\"{o}}rklund\surnameend}, \bibinfo{author}{Wim \surnamestart Martens\surnameend} \& \bibinfo{author}{Thomas \surnamestart Timm\surnameend} (\bibinfo{year}{2015}): \emph{\bibinfo{title}{Efficient Incremental Evaluation of Succinct Regular Expressions}}.
\newblock In: {\slshape \bibinfo{booktitle}{Proceedings of the 24th {ACM} International Conference on Information and Knowledge Management, {CIKM} 2015, Melbourne, VIC, Australia, October 19 - 23, 2015}}, \bibinfo{publisher}{{ACM}}, pp. \bibinfo{pages}{1541--1550}, \doi{10.1145/2806416.2806434}.

\bibitemdeclare{misc}{Cox10}
\bibitem{Cox10}
\bibinfo{author}{Russ \surnamestart Cox\surnameend} (\bibinfo{year}{2010}): \emph{\bibinfo{title}{Regular Expression Matching in the Wild: A tour of {RE2}, an efficient, production regular expression implementation}}.
\newblock \urlprefix\url{https://swtch.com/~rsc/regexp/regexp3.html}.
\newblock \bibinfo{note}{Accessed: 2025-04-26}.

\bibitemdeclare{inproceedings}{DavisCSL18}
\bibitem{DavisCSL18}
\bibinfo{author}{James~C \surnamestart Davis\surnameend}, \bibinfo{author}{Christy~A \surnamestart Coghlan\surnameend}, \bibinfo{author}{Francisco \surnamestart Servant\surnameend} \& \bibinfo{author}{Dongyoon \surnamestart Lee\surnameend} (\bibinfo{year}{2018}): \emph{\bibinfo{title}{The impact of regular expression denial of service ({ReDoS}) in practice: an empirical study at the ecosystem scale}}.
\newblock In: {\slshape \bibinfo{booktitle}{Proceedings of the 2018 26th {ACM} joint meeting on {European} {Software} {Engineering} {Conference} and {Symposium} on the {Foundations} of {Software} {Engineering}}}, pp. \bibinfo{pages}{246--256}, \doi{10.1145/3236024.3236027}.

\bibitemdeclare{inproceedings}{DavisMCSL19}
\bibitem{DavisMCSL19}
\bibinfo{author}{James~C. \surnamestart Davis\surnameend}, \bibinfo{author}{Louis~G. \surnamestart Michael\surnameend, IV}, \bibinfo{author}{Christy~A. \surnamestart Coghlan\surnameend}, \bibinfo{author}{Francisco \surnamestart Servant\surnameend} \& \bibinfo{author}{Dongyoon \surnamestart Lee\surnameend} (\bibinfo{year}{2019}): \emph{\bibinfo{title}{Why aren't regular expressions a lingua franca? {An} empirical study on the re-use and portability of regular expressions}}.
\newblock In: {\slshape \bibinfo{booktitle}{Proceedings of the {ACM} {Joint} {Meeting} on {European} {Software} {Engineering} {Conference} and {Symposium} on the {Foundations} of {Software} {Engineering}, {ESEC/SIGSOFT} {FSE} 2019, Tallinn, Estonia, August 26-30, 2019}}, \bibinfo{publisher}{{ACM}}, pp. \bibinfo{pages}{443--454}, \doi{10.1145/3338906.3338941}.

\bibitemdeclare{article}{GeladeGM12}
\bibitem{GeladeGM12}
\bibinfo{author}{Wouter \surnamestart Gelade\surnameend}, \bibinfo{author}{Marc \surnamestart Gyssens\surnameend} \& \bibinfo{author}{Wim \surnamestart Martens\surnameend} (\bibinfo{year}{2012}): \emph{\bibinfo{title}{Regular Expressions with Counting: Weak versus Strong Determinism}}.
\newblock {\slshape \bibinfo{journal}{{SIAM} J. Comput.}} \bibinfo{volume}{41}(\bibinfo{number}{1}), pp. \bibinfo{pages}{160--190}, \doi{10.1137/100814196}.

\bibitemdeclare{article}{GlaunecKM23}
\bibitem{GlaunecKM23}
\bibinfo{author}{Alexis~Le \surnamestart Glaunec\surnameend}, \bibinfo{author}{Lingkun \surnamestart Kong\surnameend} \& \bibinfo{author}{Konstantinos \surnamestart Mamouras\surnameend} (\bibinfo{year}{2023}): \emph{\bibinfo{title}{Regular Expression Matching using Bit Vector Automata}}.
\newblock {\slshape \bibinfo{journal}{Proc. {ACM} Program. Lang.}} \bibinfo{volume}{7}, pp. \bibinfo{pages}{492--521}, \doi{10.1145/3586044}.

\bibitemdeclare{article}{HolikSTV23}
\bibitem{HolikSTV23}
\bibinfo{author}{Luk{\'{a}}s \surnamestart Hol{\'{\i}}k\surnameend}, \bibinfo{author}{Juraj \surnamestart S{\'{\i}}c\surnameend}, \bibinfo{author}{Lenka \surnamestart Turonov{\'{a}}\surnameend} \& \bibinfo{author}{Tom{\'{a}}s \surnamestart Vojnar\surnameend} (\bibinfo{year}{2023}): \emph{\bibinfo{title}{Fast Matching of Regular Patterns with Synchronizing Counting (Technical Report)}}.
\newblock {\slshape \bibinfo{journal}{CoRR}} \bibinfo{volume}{abs/2301.12851}.
\newblock \eprint{2301.12851}.

\bibitemdeclare{inproceedings}{KongYCGHMY22}
\bibitem{KongYCGHMY22}
\bibinfo{author}{Lingkun \surnamestart Kong\surnameend}, \bibinfo{author}{Qixuan \surnamestart Yu\surnameend}, \bibinfo{author}{Agnishom \surnamestart Chattopadhyay\surnameend}, \bibinfo{author}{Alexis~Le \surnamestart Glaunec\surnameend}, \bibinfo{author}{Yi~\surnamestart Huang\surnameend}, \bibinfo{author}{Konstantinos \surnamestart Mamouras\surnameend} \& \bibinfo{author}{Kaiyuan \surnamestart Yang\surnameend} (\bibinfo{year}{2022}): \emph{\bibinfo{title}{Software-hardware codesign for efficient in-memory regular pattern matching}}.
\newblock In \bibinfo{editor}{Ranjit \surnamestart Jhala\surnameend} \& \bibinfo{editor}{Isil \surnamestart Dillig\surnameend}, editors: {\slshape \bibinfo{booktitle}{{PLDI} '22: 43rd {ACM} {SIGPLAN} International Conference on Programming Language Design and Implementation, San Diego, CA, USA, June 13 - 17, 2022}}, \bibinfo{publisher}{{ACM}}, pp. \bibinfo{pages}{733--748}, \doi{10.1145/3519939.3523456}.

\bibitemdeclare{inproceedings}{SmithEJ08}
\bibitem{SmithEJ08}
\bibinfo{author}{Randy \surnamestart Smith\surnameend}, \bibinfo{author}{Cristian \surnamestart Estan\surnameend} \& \bibinfo{author}{Somesh \surnamestart Jha\surnameend} (\bibinfo{year}{2008}): \emph{\bibinfo{title}{{XFA:} Faster Signature Matching with Extended Automata}}.
\newblock In: {\slshape \bibinfo{booktitle}{2008 {IEEE} Symposium on Security and Privacy {(SP} 2008), 18-21 May 2008, Oakland, California, {USA}}}, \bibinfo{publisher}{{IEEE} Computer Society}, pp. \bibinfo{pages}{187--201}, \doi{10.1109/SP.2008.14}.

\bibitemdeclare{inproceedings}{SuHLCG24}
\bibitem{SuHLCG24}
\bibinfo{author}{Weihao \surnamestart Su\surnameend}, \bibinfo{author}{Hong \surnamestart Huang\surnameend}, \bibinfo{author}{Rongchen \surnamestart Li\surnameend}, \bibinfo{author}{Haiming \surnamestart Chen\surnameend} \& \bibinfo{author}{Tingjian \surnamestart Ge\surnameend} (\bibinfo{year}{2024}): \emph{\bibinfo{title}{Towards an Effective Method of {ReDoS} Detection for Non-backtracking Engines}}.
\newblock In: {\slshape \bibinfo{booktitle}{33rd USENIX Security Symposium (USENIX Security 24)}}, pp. \bibinfo{pages}{271--288}.

\bibitemdeclare{inproceedings}{TuronovaHHLVV22}
\bibitem{TuronovaHHLVV22}
\bibinfo{author}{Lenka \surnamestart Turonov{\'{a}}\surnameend}, \bibinfo{author}{Luk{\'{a}}s \surnamestart Hol{\'{\i}}k\surnameend}, \bibinfo{author}{Ivan \surnamestart Homoliak\surnameend}, \bibinfo{author}{Ondrej \surnamestart Leng{\'{a}}l\surnameend}, \bibinfo{author}{Margus \surnamestart Veanes\surnameend} \& \bibinfo{author}{Tom{\'{a}}s \surnamestart Vojnar\surnameend} (\bibinfo{year}{2022}): \emph{\bibinfo{title}{Counting in Regexes Considered Harmful: Exposing {ReDoS} Vulnerability of Nonbacktracking Matchers}}.
\newblock In: {\slshape \bibinfo{booktitle}{31st {USENIX} Security Symposium, {USENIX} Security 2022, Boston, MA, USA, August 10-12, 2022}}, \bibinfo{publisher}{{USENIX} Association}, pp. \bibinfo{pages}{4165--4182}.

\bibitemdeclare{article}{TuronovaHLSVV20}
\bibitem{TuronovaHLSVV20}
\bibinfo{author}{Lenka \surnamestart Turonov{\'{a}}\surnameend}, \bibinfo{author}{Luk{\'{a}}s \surnamestart Hol{\'{\i}}k\surnameend}, \bibinfo{author}{Ondrej \surnamestart Leng{\'{a}}l\surnameend}, \bibinfo{author}{Olli \surnamestart Saarikivi\surnameend}, \bibinfo{author}{Margus \surnamestart Veanes\surnameend} \& \bibinfo{author}{Tom{\'{a}}s \surnamestart Vojnar\surnameend} (\bibinfo{year}{2020}): \emph{\bibinfo{title}{Regex matching with counting-set automata}}.
\newblock {\slshape \bibinfo{journal}{Proc. {ACM} Program. Lang.}} \bibinfo{volume}{4}, pp. \bibinfo{pages}{218:1--218:30}, \doi{10.1145/3428286}.

\end{thebibliography}

\end{document}